\documentclass[conference]{IEEEtran}
\IEEEoverridecommandlockouts

\usepackage[utf8]{inputenc} 
\usepackage[T1]{fontenc}    
\usepackage{hyperref}       
\usepackage{url}            
\usepackage{booktabs}       
\usepackage{amsfonts}       
\usepackage{nicefrac}       
\usepackage{microtype}      
\usepackage{lipsum}		
\usepackage{graphicx}
\usepackage{doi}

\usepackage{amsthm}
\usepackage{amsmath} 
\usepackage{amssymb}
\usepackage{array}
\usepackage{mathtools}
\usepackage{dsfont}
\usepackage{mathrsfs}
\usepackage{nicefrac}
\usepackage{soul}
\usepackage{tikz}
\usetikzlibrary{quantikz2}
\usepackage[capitalise]{cleveref}
\usepackage{subfig}

\usepackage{cite}
\usepackage{amsmath,amssymb,amsfonts}
\usepackage{algorithmic}
\usepackage{graphicx}
\usepackage{textcomp}
\usepackage{xcolor}
\def\BibTeX{{\rm B\kern-.05em{\sc i\kern-.025em b}\kern-.08em
    T\kern-.1667em\lower.7ex\hbox{E}\kern-.125emX}}

\newcommand{\ud}[1]{\, \mathrm{d}#1}

\newcommand{\Z}{\mathbb{Z}}

\newcommand{\pauli}{\mathbb{S}}

\newcommand{\Cov}{\mathbb{C}\mathrm{ov}}

\newcommand{\bs}[1]{\boldsymbol{#1}}

\newcommand{\bK}{\bs{K}}
\newcommand{\bk}{\bs{k}}

\newcommand{\btheta}{\bs{\theta}}

\def\hmath$#1${\texorpdfstring{{\rmfamily\textit{#1}}}{#1}}

\DeclareMathOperator{\tr}{tr}

\DeclareMathOperator{\sspan}{span}

\newtheorem{theorem}{Theorem}

\newtheorem{proposition}{Proposition}

\title{Probabilistic Representation of Commutative Quantum Circuit Models}

\author{Richard Yu$^1$ \and Jorge Ramirez$^2$ \and Elaine Wong$^3$}
\date{%
    $^1$Georgia Tech, yohan.richard.yu@gmail.com \\[1ex]%
    $^2$Oak Ridge National Laboratory, ramirezosojm@ornl.gov \\[1.1ex]%
    $^3$Oak Ridge National Laboratory, wongey@ornl.gov
}

\begin{document}

\title{Probabilistic Representation of Commutative Quantum Circuit Models\\
\thanks{This manuscript has been authored by UT-Battelle, LLC under Contract No. DE-AC05-00OR22725 with the U.S. Department of Energy. The United States Government retains and the publisher, by accepting the article for publication, acknowledges that the United States Government retains a non-exclusive, paid-up, irrevocable, world-wide license to publish or reproduce the published form of this manuscript, or allow others to do so, for United States Government purposes. The Department of Energy will provide public access to these results of federally sponsored research in accordance with the DOE Public Access Plan (http://energy.gov/downloads/doe-public-access-plan).}
}

\author{\IEEEauthorblockN{Richard Yu}
\IEEEauthorblockA{\textit{School of Mathematics} \\
\textit{Georgia Institute of Technology}\\
Atlanta, GA 30332\\
ryu83@gatech.edu}
\and
\IEEEauthorblockN{Jorge Ramirez Osorio}
\IEEEauthorblockA{\textit{Computer Science and Mathematics Division}
 \\
\textit{Oak Ridge National Laboratory}\\
Oak Ridge, TN 37830\\
ramirezosojm@ornl.gov}
\and
\IEEEauthorblockN{Elaine Wong}
\IEEEauthorblockA{\textit{Computer Science and Mathematics Division} \\
\textit{Oak Ridge National Laboratory}\\
Oak Ridge, TN 37830\\
wongey@ornl.gov}
}

\maketitle

\begin{abstract}
In commuting parametric quantum circuits, the Fourier series of the pairwise fidelity can be expressed as the characteristic function of random variables. Furthermore, expressiveness can be cast as the recurrence probability of a random walk on a lattice. This construction had previously been applied to the group composed only of Pauli-$Z$ rotations. In this paper, we generalize this probabilistic strategy to any commuting set of Pauli operators. To this end, we can leverage an algorithm that uses the tableau representation of Pauli strings to yield a unitary from the Clifford group that, under conjugation, simultaneously diagonalizes our commuting set of Pauli rotations. Furthermore, we fully characterize the underlying distribution of the random walk using stabilizer states and their basis state representations. This would allow us to tractably compute the lattice volume and variance matrix used to express the frame potential. Together, this demonstrates a scalable strategy to calculate the expressiveness of parametric quantum models.
\end{abstract}

\begin{IEEEkeywords}
quantum computing, quantum circuit model, commutative circuit, probabilistic representation, quantum machine learning
\end{IEEEkeywords}

\section{Introduction}

\subsection{Motivation}

Designing quantum circuits that can produce a wide range of candidate solutions for a particular problem class is a persistent challenge in the development of Parametric Quantum Circuits (PQCs) for quantum machine learning.
The range of quantum states or unitaries a variational circuit can produce is known as its \emph{expressiveness}~\cite{Du2020}, \emph{expressibility}~\cite{Sim2019, Holmes2022}, or capacity~\cite{Haug2021}. 
The greater the expressiveness, the more likely the circuit is able to produce a good solution to a given problem which is useful in identifying better choices of architecture.
Expressiveness is related to computational complexity~\cite{Morales2020} as well as trainability~\cite{Holmes2022}, and has been quantified in various ways, including the rank of the Jacobian matrix~\cite{Haghshenas2022} or quantum Fisher information matrix~\cite{Larocca2023}, distance to a quantum $t$-design~\cite{Sim2019, Holmes2022}, or closeness of a pairwise fidelity distribution to an ideal distribution~\cite[Section 3.1.1]{Sim2019}.
Following the work of \cite{ramirez2024expressiveness}, we conceptualize the expressiveness of a quantum circuit in terms of the \emph{frame potential} of the circuit. This approach has previously been used in \cite{Sim2019}, and is based on the observation that frame potentials can be used to estimate the non-uniformity of the set of states generated by a PQC, namely, they can be understood as the ability of a PQC to traverse the Bloch sphere. Efficient computation of frame potentials, thus provides a means to assess the complexity of a PQC architecture. 

However, the frame potential of a given PQC is known to be intractable to compute due to the exponential scaling of qubit systems. We continue to develop our understanding of this intractability by expanding on the work from~\cite{ramirez2024expressiveness}, and present a scalable and theoretically rigorous  strategy for computing the asymptotic behavior of the frame potential of any PQC defined by a commuting set of Pauli operators. 

\subsection{Relevant Definitions and Main Goal}

Throughout this paper, we will use bold lower-case symbols to denote vectors, as in $\bs{u}, \bs{k}$ and $\bs{\theta}$, with its individual entries denoted as $\bs{u} = (u^{(1)}, u^{(2)}, \dots)$ or $\bs{\theta} = (\theta_1,\theta_2,\dots)$ depending if the vector represents a row or a column, respectively. Matrices are bold capitalized as in $\bs{K}$ or $\bs{A}$. Random variables are always capitalized as in $K$. The index $j$ is only used to count over gates $\{1,\dots,N\}$, and the index $x$ is always in $\{0,\dots,2^n-1\}$.

We consider PQCs defined by unitary operators of the form
\begin{equation}\label{eq:defU}
    U(\btheta) := \prod_{j=1}^N e^{i \theta_j H_j}, \quad \theta_j \in [-\pi, \pi],
\end{equation}
where $n$ is the number of qubits, $N$ the number of gates, $n \leq N \leq 2^n$, and $\{H_j\}_{j=1}^N$ is a commuting set of Pauli strings. The frame potential of $U$ is
\begin{equation}\label{def:framePot}
    \mathcal{F}_U(t) := \frac{1}{(2 \pi)^{2N}} \int_{[-\pi,\pi]^{2N}} F_U(\btheta,\btheta')^{t} \ud \btheta \ud \btheta', \quad t \geq 0,
\end{equation}
where $F_U$ is the fidelity of $U$,
\begin{equation}\label{eq:defFU}
    F_U(\btheta,\btheta') := \left| \langle 0 |^{\otimes n} U(\btheta)^\dagger  U(\btheta') |0\rangle^{\otimes n} \right|^{2}.
\end{equation}

Here, the fidelity $F_U \in [0,1]$ quantifies the similarity between states produced by circuits with two different parameter settings. This means that if we fix $t>0$, the frame potential can be understood as quantifying the average similarity between all pairs of states the circuit can produce.  A relatively small frame potential indicates that most producible states are dissimilar to each other, which intuitively corresponds to large expressiveness or complexity. Since $\{H_j\}_{j=1}^N$ is a commuting set, we can write fidelity as
\begin{equation}
F_U (\btheta, \btheta') = |\langle0|^{\otimes n} U(\btheta - \btheta') |0\rangle^{\otimes n}|^2 = |f_U (\btheta - \btheta')|^2,
\end{equation} 
where 
\begin{equation} \label{eq:deffU}
    f_U (\btheta) = \langle 0 |^{\otimes n} U (\btheta) |0 \rangle^{\otimes n}.
\end{equation}

The \textbf{goal of this paper} is to derive a probabilistic representation of $f_U$ and characterize the asymptotic behavior of $\mathcal{F}_U (t)$ as $t \to \infty$ when $\{H_j\}_{j = 1}^N$ is \textbf{any} commuting set of Pauli operators. 

Our strategy is to generalize the probabilistic method developed in~\cite{ramirez2024expressiveness} which applies exclusively to the case where each string $H_j$ is composed only of identity and Pauli-$Z$ operators. Namely, the list $\{H_j\}_{j = 1}^N$ was chosen from the Cartan subalgebra  $S = \{I, Z\}^{\otimes n} - \{I^{\otimes n}\}$ of the $n$-qubit Pauli group $\pauli_n$. This group can be partitioned into a total of  $2^n + 1$ such Cartan subalgebras. In this paper, we derive the probabilistic representation of $f_U$ to  PQCs with Hamiltonians chosen from any of the remaining subalgebras, thereby generalizing the work to all of $\pauli_n$. Notably, using the theory of stabilizer sets is what adds novelty to and delineates our work from the probabilistic foundations that were set in~\cite{ramirez2024expressiveness}.

From now on, we restrict to operators $U$ of the form \cref{eq:defU} where $H_1, \dots, H_N$ are commuting Pauli operators in $\pauli_N$. These have integer spectra, and since they are simultaneously diagonalizable~\cite{vandenBerg2020circuitoptimization}, there exists a unitary $W$ in the Clifford group such that $\Lambda_j = W H_j W^\dagger \in S$ is a diagonal matrix for all $1 \leq j \leq N$. Therefore, 
\[f_U (\btheta) = \langle0|^{\otimes n} W^\dagger \operatorname{exp}(i (\theta_1 \Lambda_1 + \dots + \theta_N \Lambda_N)) W |0\rangle^{\otimes n}.\] 
Using this reduction, we can assume without loss of generality that the operator $U$ is the matrix exponential of the sum of diagonal matrices and we can write,
\begin{align}
    U(\btheta) &= \exp(i \theta_1 \Lambda_1 + \dots + i \theta_N \Lambda_N), \label{eq:defULambdas}\\
    f_U(\btheta) &= \langle \psi_0| U(\btheta) |\psi_0\rangle \text{ where } |\psi_0\rangle = W |0\rangle^{\otimes n}.\label{eq:fUW}
\end{align}
Each of the diagonal operators $\Lambda_j$ can be represented as a Pauli string of only $Z$ and identity operators, so parts of the formulation in \cite{ramirez2024expressiveness} can be directly applied. The generalization to all of $\pauli_n$, which is the main contribution of this paper, is in characterizing the effect that $W$ has on the probabilistic representation of the frame potential.

\subsection{Organization of Paper}

\cref{sec:ProbRep} formulates the quantum expectation in \cref{eq:fUW} as a characteristic function of a random variable on the spectrum of $U$. Characterizing the distribution of this random variable is algorithmically fleshed out in \cref{sec:K}, using properties of stabilizer states. An example application is presented in \cref{sec:example}. In \cref{sec:noncommutative}, we discuss two potential approaches to rigorously generalize this work to all sets of Pauli operators.

\section{The Probabilistic Representation}\label{sec:ProbRep}
Following~\cite{ramirez2024expressiveness}, we can interpret the quantum expectation $f_U(\btheta)$ in \cref{eq:fUW} from a probabilistic perspective. First note that because it is a positive-definite function, by Bochner's theorem, $f_U$ is the characteristic function of some random variable $K$. For the construction of $K$, we index the rows and columns of each $\Lambda_j$ by $x=0,\dots,2^n-1$ and define
\begin{equation}
    \bk_x := (k_x^{(1)}, \dots, ..., k_x^{(N)}), \quad k_x^{(j)} = (\Lambda_j)_{x,x}.
\end{equation}
Then, we can extract the diagonal as \[\operatorname{diag}(\theta_1 \Lambda_1 + \dots + \theta_N \Lambda_N) = \left(\btheta \cdot \bk_0\ \dots\ \btheta \cdot \bk_{2^n - 1}\right)^\intercal.\] Let $\bs{u}_x \in \mathbb{Z}_2^n$ denote the integer $x$ written in its $n$-bit representation as a vector. Then, we find that 
\begin{align}
f_U(\btheta) &= \tr\left(|\psi_0\rangle \langle\psi_0| e^{i (\theta_1 \Lambda_1 + \dots + \theta_N \Lambda_N)}\right) \nonumber \\ &= \sum_{x=0}^{2^n - 1} |\langle\psi_0|\bs{u}_x\rangle|^2 e^{i \btheta \cdot \bk_{x}} = \mathbb{E}[e^{i \btheta \cdot \bk_X}], \nonumber
\end{align}
where $X$ is a random variable taking values in $\{0, ..., 2^n - 1\}$ with probability mass function
\begin{equation}\label{eq:pmfX}
    \mathbb{P}(X = x) = |\langle \psi_0|\bs{u}_x\rangle|^2.
\end{equation}
Note that $\sum_{x=0}^{2^n - 1}  |\langle \psi_0 | \bs{u}_x \rangle|^2 = 1$ because $|\psi_0\rangle$ is a normalized quantum state. We will think of $\bs{k}_X$ as a random row from
\begin{equation}\label{def:bK}
    \bK := \begin{pmatrix}
         -& \bk_0 & -   \\
        & \vdots &\\
         -& \bk_{2^n - 1}& - 
    \end{pmatrix} \in  \Z^{2^n \times N}.
\end{equation}
If we define the random variable $K = \bk_X$, then the probability distribution of $K$ can be computed by tallying the amplitudes of the indices where a given $\bs{k}$ is repeated in $\bs{K}$,  
\begin{equation}\label{eq_pmfK}
\mathbb{P}(K = \bk) = \sum_{x:\bk_x = \bk} |\langle \psi_0 | \bs{u}_x \rangle|^2.
\end{equation}

Theorem 1 in~\cite{ramirez2024expressiveness} shows that the frame potential in \cref{def:framePot} can be written as $\mathcal{F}_U(t) = \mathbb{P}(W_t = 0)$ where $\{W_t\}_{t=0}^{\infty}$ is a random walk on $\Z^n$ with increments given by differences of independent random variables distributed as \cref{eq_pmfK}. The central limit theorem gives the following approximation to $\mathcal{F}$,
\begin{equation}\label{eq:tildeFU}
\Tilde{\mathcal{F}}_U (t) = \frac{V_U}{\sqrt{(4 \pi t)^N \operatorname{det}(\Cov(K))}},
\end{equation}
where $V_U$ is the volume of the lattice in $\Z^n$ generated by the random walk $W_t$. Specifically, $|\mathcal{F}_U(t) - \Tilde{\mathcal{F}}_U (t)| = O(t^{-(1+N/2)})$ as $t \to \infty$ (see~\cite[Theorem 2]{ramirez2024expressiveness}). In particular, the range and distribution of $K$ depends on the algebraic structure of $\Lambda_j$, as its mean and covariance matrix in \cref{eq:tildeFU} are
\begin{align}
    \mathbb{E} K^{(j)} &= \sum_{x=0}^{2^n - 1} k_x^{(j)} |\langle\psi_0|\bs{u}_x\rangle|^2 = \langle \psi_0 | \Lambda_j | \psi_0 \rangle, \\
    \Cov(K)_{i,j} &= \langle \psi_0| \Lambda_i \Lambda_j |\psi_0 \rangle - \langle \psi_0 | \Lambda_i | \psi_0 \rangle \langle \psi_0 | \Lambda_j | \psi_0 \rangle.
\end{align}

In the case of Pauli strings containing only the $Z$ and identity operator, it is proven in~\cite{ramirez2024expressiveness} that the entries of $K$ are independent, and identically distributed Bernoulli random variables. We now derive the corresponding result for any Pauli string. 

\subsection{Expression for the Distribution of \texorpdfstring{$K$}{K}}

We encode $U$ in a binary matrix $\bs{A} \in \{0,1\}^{N \times n}$ by making its $(j,m)$-th entry $A_{j}^{(m)}$ equal to 1 if the $m$-th operator of the Pauli string representation of $\Lambda_j$ is a $Z$, and 0 otherwise. We will use $\bs{A}$ to determine the rows of $K$ and its distribution.

\begin{theorem}\label{thm:pmfK1}
Let $\bs{A}$ be as above. The random variable $K$ takes values on a subset of $\{-1,1\}^N$. If we write a general vector in such set as $(-1)^{\boldsymbol{b}}$ for some $\boldsymbol{b} \in \mathbb{Z}_2^N$, then,
\begin{equation}\label{eq_kDistribution} 
\mathbb{P}(K = (-1)^{\bs{b}}) = \sum_{x: \bs{A} \bs{u}_x=\bs{b}} |\langle \psi_0 | \bs{u}_x \rangle|^2. 
\end{equation}
\end{theorem}
\begin{proof}
Because each $\Lambda_j$ is a Pauli string of $I$ and $Z$ gates, the diagonal of $\Lambda_j$ can be expressed as  
\[
\operatorname{diag}(\Lambda_j) = (-1)^{\begin{psmallmatrix} 0 \\ A_{j}^{(1)}\end{psmallmatrix}} \otimes (-1)^{\begin{psmallmatrix}0 \\ A_{j}^{(2)} \end{psmallmatrix}} \otimes \dots \otimes (-1)^{\begin{psmallmatrix} 0 \\ A_{j}^{(n)} \end{psmallmatrix}}.
\]
We find $K_x^{(j)} = (-1)^p$ with 
$
p = \bigoplus_{m=1}^n A_j^{(m)} u_x^{(m)}
$
where $\oplus$ denotes addition mod 2. This means that any row $\bs{k}_x$ of $K$ can be expressed as $\bs{k}_x = (-1)^{\bs{A} \bs{u}_x}$. Namely $K = \bs{k}_X = (-1)^{\bs{b}}$ if and only if $\bs{A} \bs{u}_X = \bs{b}$. \cref{eq_kDistribution} follows from \cref{eq_pmfK}. \end{proof}

\section{Characterization of the Distribution of \texorpdfstring{$K$}{K}}\label{sec:K}

We now turn our attention to characterizing the possible values of $|\langle \psi_0 | \bs{u}_x \rangle|^2$ needed to compute the distribution of $K$ in \cref{eq_kDistribution}. These are a property of the base state, which recalling \cref{eq:fUW}, is given by $|\psi_0\rangle = W |0\rangle^{\otimes n}$.

\subsection{Representations of Stabilizer States and Their Amplitudes}

Since $W$ is in the Clifford group, it can be represented as some combination of Clifford gates. This means that $|\psi_0\rangle = W |0\rangle^{\otimes n}$ is a stabilizer state, namely, a quantum state that is obtainable from $|0\rangle^{\otimes n}$ by applying only just Controlled-NOT gates, Hadamard gates, and Phase gates. See~\cite{Aaronson2004ImprovedSO}. From~\cite[Theorem 9]{StabilizerStateProperties}, for any $n$-qubit stabilizer state $|\psi_0\rangle$, there exist vectors $\bs{c}, \bs{t} \in \mathbb{Z}_2^n$, a symmetric binary matrix $\bs{Q} \in \Z_2^{n\times n}$ and $\bs{R} \in \Z_2^{n\times r}$ of rank $r$ for some $r \leq n$, such that 
\begin{equation}\label{eq_stabilizerStateFormula} 
\begin{split}
|\psi_0\rangle &= \frac{1}{2^{r/2}} \sum_{\bs{z} \in \mathbb{Z}_2^r} i^{2\bs{c}^\intercal f(\bs{z}) + f(\bs{z})^\intercal \bs{Q} f(\bs{z})} |f(\bs{z}) \rangle,\\
f(\bs{z}) &= \bs{Rz} + \bs{t}, \quad \bs{z} \in \Z_2^r.
\end{split}
\end{equation}
See~\cite{DehaeneStabilizerStates} for an efficient method to compute for $\bs{R}$ and $\bs{t}$. Let  
$x \in 0, \dots, 2^n-1$. Then, substituting \eqref{eq_stabilizerStateFormula} gives
\begin{align}
|\langle \psi_0 | \bs{u}_x \rangle|^2  
&= \left| \left\langle \bs{u}_x \middle| \frac{1}{2^{r/2}}  \sum_{\bs{z} \in \mathbb{Z}_2^r} i^{2\bs{c}^\intercal f(\bs{z}) + f(\bs{z})^\intercal \bs{Q} f(\bs{z})} |f(\bs{z}) \right\rangle \right|^2 \nonumber\\
&= \begin{cases}
    \frac{1}{2^r} & \bs{u}_x \in \operatorname{range}(f) \\
    0 & \operatorname{otherwise}
\end{cases}. \label{eq:amplitude}
\end{align}

It follows by \cref{eq:pmfX} that $\mathbb{P}(X = x)$ is either zero or $2^{-r}$. Namely, $X$ is uniformly distributed on some subset of $\{0,\dots,2^n-1\}$ containing $2^r$ elements. We thus arrive at a further specification of the distribution of $K$.

\begin{proposition}\label{prop:pmfK2}
    Let $\bs{A}$ be as in \cref{thm:pmfK1} and $\bs{R}$, $f$ as in \cref{eq_stabilizerStateFormula}. Then for $\bs{b} \in \Z_2^N$,
        \begin{equation}\label{eq_pauliDistributionK}
        \mathbb{P}(K = (-1)^{\bs{b}}) = 
        \begin{cases} 
            2^{-\operatorname{rank}(\bs{AR})} ,& \bs{b} \in \operatorname{range}(\bs{A} f), \\ 
            0, & \operatorname{otherwise}. 
        \end{cases}
        \end{equation}
\end{proposition}
\begin{proof}
    Combining \cref{eq_kDistribution,eq:amplitude} gives 
        \begin{align*}
        \mathbb{P}(K = (-1)^{\bs{b}}) &= \sum_{x : \bs{A} \bs{u}_x \nonumber = \bs{b}} |\langle \psi_0 | \bs{u}_x\rangle|^2 \nonumber \\ 
        &= \frac{1}{2^r}\#\{  \bs{z} \in \mathbb{Z}_2^r : \bs{A}f(\bs{z}) = \bs{b} \} \\
        &= \frac{1}{2^r} \begin{cases} 
            2^{\operatorname{dim}(\operatorname{null}(\bs{AR}))}, & \bs{b} \in \operatorname{range}(\bs{A} f), \\ 
            0, & \operatorname{otherwise}.
        \end{cases}
        \end{align*}
    The results follows by noting that since $\bs{A R} \in \mathbb{Z}_2^{N \times r}$ and it is assumed that $N \geq n \geq r$, then $\operatorname{dim} (\operatorname{null}(\bs{AR})) = r - \operatorname{rank}(\bs{A R})$. \end{proof}

Note that \cref{prop:pmfK2} generalizes Thm. 3 in~\cite{ramirez2024expressiveness}. There, the probability mass function of $K$ was simply
$\mathbb{P}(K = (-1)^{\bs{b}}) = 2^{-\operatorname{rank}(\bs{A})}$,
which corresponds to \cref{eq_pauliDistributionK} in the case where the state $|\psi_0\rangle$ is maximally mixed, namely $W = H^{\otimes n}$. This assumption sets $r = n$, and $\operatorname{rank}(\bs{A R}) = \operatorname{rank}(\bs{A})$.

In summary, according to \cref{thm:pmfK1} and \cref{prop:pmfK2}, the possible values of the random variable $K$ are $2^{\operatorname{rank}(\bs{AR})}$ distinct rows of the matrix $\bs{K}$ corresponding to $2^r$ indices $x \in \{0,\dots,2^n-1\}$. Each such row is thus repeated $2^{r - \operatorname{rank}(\bs{AR})}$ times in $\bs{K}$. The set of indices $x$ of these rows in $\bs{K}$ is in one-to-one correspondence with the set of vectors $\bs{u}_x$ that can written as $\bs{u}_x = f(\bs{z}) = \bs{R} \bs{z} + \bs{t}$ for some $\bs{z} \in \Z_2^r$,  which in turn is equivalent to the set of basis states $| \bs{u}_x \rangle$ such that $|\langle \psi_0 | \bs{u}_x \rangle|^2 > 0$. The \textit{support} of $K$ can thus be specified by the set of basis states
\begin{equation}\label{eq:defJK}
    J_K = \{|\bs{R} \bs{z} + \bs{t}\rangle : \bs{z} \in \mathbb{Z}_2^r\}.
\end{equation}
Thus, to complete the understanding of how $K$ is distributed, it is enough to study the construction of $J_K$. 

\subsection{Generating the Support}
Tableaus are employed by~\cite{Aaronson2004ImprovedSO} to represent stabilizer states, and there are measurement procedures on these tableaus that will allow us to identify $\bs{R}$ and $\bs{t}$ in \cref{eq:defJK}. 

Denote by $\operatorname{H}(a)$ the Pauli string of $n$ qubits with all positions set to $I$, except the $a$-th position is the Pauli-$Z$ matrix. Define $M_{q,0}$ and $M_{q, 1}$ to be the projection operator of the Pauli string $\operatorname{H}(q)$'s positive and negative eigenspaces respectively, namely 
\begin{equation}\label{eq:defM}
    M_{q, 0} = (I + \operatorname{H}(q))/2, \quad M_{q, 1} = (I - \operatorname{H}(q))/2.
\end{equation}
We will denote by $\tau(|\psi_{0}\rangle)$ the tableau representation of the stabilizer state $|\psi_{0}\rangle$. See \cite{Aaronson2004ImprovedSO}. A stabilizer state has the property that it can be uniquely represented by its Pauli stabilizer group which has $n$ generators, each being a Pauli string. These are encoded in the tableau $\tau(|\psi_0\rangle)$ in the following manner:
\begin{table}[!hbt]
\centering
\resizebox{\columnwidth}{!}{%
$\left(
\begin{tabular}[c]{ccc|ccc|c}%
$x_{11}$ & $\cdots$ & $x_{1n}$ & $z_{11}$ & $\cdots$ & $z_{1n}$ & $r_{1}$\\
$\vdots$ & $\ddots$ & $\vdots$ & $\vdots$ & $\ddots$ & $\vdots$ & $\vdots$\\
$x_{n1}$ & $\cdots$ & $x_{nn}$ & $z_{n1}$ & $\cdots$ & $z_{nn}$ & $r_{n}$\\\hline
$x_{\left(  n+1\right)  1}$ & $\cdots$ & $x_{\left(  n+1\right)  n}$ &
$z_{\left(  n+1\right)  1}$ & $\cdots$ & $z_{\left(  n+1\right)  n}$ &
$r_{n+1}$\\
$\vdots$ & $\ddots$ & $\vdots$ & $\vdots$ & $\ddots$ & $\vdots$ & $\vdots$\\
$x_{\left(  2n\right)  1}$ & $\cdots$ & $x_{\left(  2n\right)  n}$ &
$z_{\left(  2n\right)  1}$ & $\cdots$ & $z_{\left(  2n\right)  n}$ & $r_{2n}$
\end{tabular}\right),
$%
}
\end{table}

The bits $x_{(n+i)j}, z_{(n+i)j}$ in the bottom $n$ rows of $\tau(|\psi_0\rangle)$ represent the $j$-th Pauli operator of the $i$-th generator's string with $00 \rightarrow I$, $01 \rightarrow Z$, $10 \rightarrow X$, and $11 \rightarrow Y$. The top $n$ rows encode the strings of the "destabilizer" generators, which together with the bottom $n$ Pauli strings generate the entire Pauli group. Define $X(|\psi_{0}\rangle)$ and $Z(|\psi_{0}\rangle)$ to be the $X$ and $Z$ matrix portions of the tableau (left and middle blocks respectively), and $S(|\psi_{0}\rangle)$ the sign column on the right. Also, let $\overline{X}(|\psi_{0}\rangle)$ be the square matrix corresponding to be the last $n$ rows of  $X(|\psi_{0}\rangle)$.

From~\cite{Aaronson2004ImprovedSO}, the tableau formalism is sufficient for us to efficiently simulate any quantum circuit made from Clifford gates. Further, this simulation procedure will yield states that are reachable from measurements of stabilizer states, which correspond precisely to those states where $|\langle \psi_0 | \bs{u}_x \rangle|^2 \neq 0$. We then can prove the following theorem, which completes the characterization of the random variable $K$.

\begin{theorem}\label{thm:JK}
    Let $|\psi_0\rangle$ be the stabilizer state of $U$ as in \cref{eq:fUW}, and let $x_0 \in \{0,\dots,2^n-1\}$ be such that $|\langle \psi_0 |\bs{u}_{x_0} \rangle|^2 \neq 0$. Then the set $J_K$ in \cref{eq:defJK} can be obtained by making $\bs{t} = \bs{u}_{x_0}$ and $\bs{R} \in \Z_2^{n \times r}$ equal to a matrix whose column space spans the row space $\overline{X}(|\psi_0\rangle)$.  Moreover, $r =\text{rank}(\overline{X}(|\psi_0\rangle))$.
\end{theorem}
\begin{proof}
Let $M_{a,s}$, $1\leq a \leq n$ and $s \in \{0,1\}$, be the projector operators in \cref{eq:defM}. From~\cite{Aaronson2004ImprovedSO}, when $M_{a,s}$ is performed with respect to the stabilizer state $|\psi_0\rangle$, there are only two possibilities:
\begin{enumerate}
\item The measurement outcome is "random", namely
\[
\langle\psi_0|M_{a, 0}^\dagger M_{a, 0} |\psi_0\rangle = \langle\psi_0|M_{a, 1}^\dagger M_{a, 1} |\psi_0\rangle = \tfrac{1}{2}.
\]
In this case, denote the measurement outcome to be $s \in \{0, 1\}$. Then the resulting state after measurement is
\[
\frac{M_{a, s} |\psi_0\rangle}{\sqrt{\langle\psi_0|M_{a, s}^\dagger M_{a, s} |\psi_0\rangle}} = \sqrt{2} M_{a, s} |\psi_0\rangle.
\]
\item The measurement outcome $s \in \{0, 1\}$ is determinate and $M_{a, s} |\psi_0\rangle = |\psi_0\rangle$.
\end{enumerate}

Suppose $|\bs{u}\rangle = |\bs{u}_{x_0}\rangle$ is a basis state such that $|\langle \psi_0 |\bs{u} \rangle|^2 \neq 0$ as in the statement of the theorem. We write $|\bs{u}\rangle$ as a binary sequence $|u^{(1)}\dots u^{(n)}\rangle$. Since $|\bs{u}\rangle$ is a measurable outcome from $|\psi_0\rangle$, we know that 
\begin{equation}\label{eq_uxpsi0n}
|\bs{u}\rangle = 2^{r/2} M_{n, u^{(n)}} M_{n-1, u^{(n-1)}} \dots M_{1, u^{(1)}} |\psi_0\rangle, 
\end{equation}
where $r$ is some positive integer that represents, for now, the number of measurements with "random" outcomes. It follows that there \cref{eq_uxpsi0n} can be written with only $r$ qubit positions $\{q_1, q_2, \dots, q_r\} \subset \{1, \dots, n\}$ as 
\begin{equation}\label{eq_rProjectionsX}
|\bs{u}\rangle = 2^{r/2} M_{q_r, u^{(q_r)}} M_{q_{r-1}, u^{(q_{r-1})}} \dots M_{q_1, u^{(q_1)}} |\psi_0\rangle.
\end{equation}
Define $|\psi_t\rangle = 2^{t/2} M_{q_t, u^{(q_t)}} \dots M_{q_1, u^{(q_1)}} |\psi_0\rangle$ for $1\leq t \leq r$. For a fixed $t$, we want to identify the action on the tableau $\tau(|\psi_{t-1}\rangle)$ that mirrors the projection operator $M_{q_t, u^{(q_t)}}$ applied to the state $|\psi_{t-1}\rangle$. 

We know that the measurement of qubit $q_t$ with respect to $|\psi_{t-1}\rangle$ yields a random outcome. According to~\cite{Aaronson2004ImprovedSO}, this is the case if and only if 
there exists $p_t \in \{n+1, \dots, 2n\}$ such that $X(|\psi_{t-1}\rangle)_{p_t}^{(q_t)} = 1$. In the tableau representation, when $M_{q_t, u^{(q_t)}}$ is applied, we can use the following algorithm to obtain $\tau(|\psi_t\rangle)$ from $\tau(|\psi_{t-1}\rangle)$: 
\begin{enumerate}
    \item For every $i \in \{1, \dots, 2n\}$ such that $i \neq p_t$ and $X(|\psi_{t-1}\rangle)_i^{(q_t)} = 1$, add row $i$ with row $p_t$ and replace row $i$ with the new row.
    \item Set the row $\tau(|\psi_{t}\rangle)_{p_t - n} = \tau(|\psi_{t-1}\rangle)_{p_t}$. 
    \item Set the row $\tau(|\psi_{t}\rangle)_{p_t} = 0$ except for the entry $Z(|\psi_{t}\rangle)_{p_t}^{(q_t)} = 1$.
    \item Assign $S(|\psi_{t}\rangle)_{p_t}$ the value $u^{(q_t)}$.
\end{enumerate}

We now make some observations on the rank of $X(|\psi_t\rangle)$. In step 1, we remove $\overline{X}(|\psi_{t-1}\rangle)_{p_t}$ from any other row in $\overline{X}(|\psi_{t-1}\rangle)$ which can be expressed as a linear sum of other row vectors that includes $\overline{X}(|\psi_{t-1}\rangle)_{p_t}$. Therefore, the rank is not altered by step 1. In step 3, the row $\overline{X}(|\psi_{t-1}\rangle)_{p_t}$ is set to zero, and so, complementing this with step 1, we find that \[
\operatorname{rank}(\overline{X}(|\psi_t\rangle)) = \operatorname{rank}(\overline{X}(|\psi_{t-1}\rangle)) - 1.
\]

Measurements done on any qubit for a single basis state must be determinate. This implies that $\overline{X}(|\bs{u}\rangle)$ is identically zero. Since each measurement with a random outcome decreases the rank of $\overline{X}(|\psi_t\rangle)$ by one, then the number of measurements is $r = \operatorname{rank}(\overline{X}(|\psi_0\rangle))$. Moreover, from the measurement procedure, $X(|\psi_0\rangle)_{p_t}$ must be linearly independent of the set $\{X(|\psi_0\rangle)_{p_s} \}_{s=1}^{t-1}$. Therefore, $\{X(|\psi_0\rangle)_{p_t}\}_{t=1}^r$ is a linearly independent set of rows.  Furthermore, since each row $p_t$ is copied to row $p_t - n$ during the $r$ applications of Step~2, then the row space of $\overline{X}(|\psi_0\rangle)$ is
\begin{equation*}
\sspan\{X(|\bs{u}\rangle)_{p_t - n} \}_{t=1}^r  
= \sspan \{X(|\psi_0\rangle)_{p_t}\}_{t=1}^r.
\end{equation*}

We now identify the new basis state measured if one of the $r$ measurement outcomes that produced $|\bs{u}\rangle$ in \cref{eq_rProjectionsX} is negated. This will provide us a means to generate all other measurement possibilities. Without loss of generality, suppose we replace the $r$-th operator $M_{q_r, u^{(q_r)}}$ in \cref{eq_rProjectionsX}  with $M_{q_r, u^{(q_r)} \oplus 1}$ (since the projection operators are commutative, any other index can be chosen and then the corresponding projection rotated to the last position). This would yield a new basis state $|\bs{v}\rangle$ which is also a measurable outcome from $|\psi_0\rangle$. From the measurement procedure with random outcomes, this means that $\tau(|\bs{v}\rangle)$ would equal $\tau(|\bs{u}\rangle)$ except that $S(|\bs{v}\rangle)_{p_r} = S(|\bs{u}\rangle)_{p_r} \oplus 1$, as the only difference was the sign assignment at step 4 on the very last projection operator. 

We can calculate the basis state $|\bs{v}\rangle$ given its tableau representation. We know $\overline{X}(|\bs{v}\rangle)$ is identically zero. If we were to perform a measurement on a qubit $q$, it would be determinate. From the determinate measurement procedure described in~\cite{Aaronson2004ImprovedSO}, the measurement outcome $v^{(q)}$ is equal to the last entry of the summation of all the rows $\tau(|\bs{v}\rangle)_{i + n}$ over $1 \leq i \leq n$ such that $X(|\bs{v}\rangle)_{i}^{(q)} = 1$.  In other words, since $\overline{X}(|\bs{v}\rangle)$ is identically zero, we find that
\[ 
v^{(q)} = \bigoplus_{1 \leq i \leq n: X(|\bs{v}\rangle)_i^{(q)} = 1} S(|\bs{v}\rangle)_{i + n}.
\] 
If we then consider the measurement of the same qubit $q$ for $\tau(|\bs{u}\rangle)$, then we could rewrite the expression as 
\[
v^{(q)} = u^{(q)} \oplus \begin{cases}
    1, & \text{if } X(|\bs{u}\rangle)_{p_r - n}^{(q)} = 1, \\
    0, & \operatorname{otherwise}.
\end{cases}
\]
Namely, $|\bs{v}\rangle = |\bs{u} \oplus X(|\bs{u}\rangle)_{p_r - n}\rangle$. This holds true when $p_r$ is replaced by any $p_t$ in Eq. \eqref{eq_rProjectionsX}. Moreover, for any sequence of indices $\{t_1, \dots, t_m : 1 \leq m \leq r\} \subset \{1, \dots, r\}$, we will have
\[
|\langle \psi_0 | \bs{u} \oplus X(|\bs{u}\rangle)_{p_{t_1} - n} \oplus \dots \oplus X(|\bs{u}\rangle)_{p_{t_m} - n} \rangle |^2 \neq 0.
\]
This indicates that $|\langle \psi_0 | \bs{u} \oplus \bs{v}\rangle|^2 \neq 0$ for any $\bs{v}$ in $\operatorname{span} \{X(|\bs{u}\rangle)_{p_t - n} : 1 \leq t \leq r\}$, which is identical to the row space of $\overline{X}(|\psi_0\rangle)$.   

Finally, let $\bs{R} \in \Z_2^{n\times r}$ be the nonsingular matrix whose columns span the row space of $\overline{X}(|\psi_0\rangle)$. Writing $\bs{v} = \bs{R} \bs{z}$ for $\bs{z} \in \mathbb{Z}_2^r$ and $\bs{t} = \bs{u}$, we get the desired characterization of the set $J_K$ of all basis states with nonzero amplitudes in $|\psi_0\rangle$.   \end{proof} 

We observe that we can use the simulation procedure described in~\cite{Aaronson2004ImprovedSO} to find an initial measurable outcome $|\bs{u}_{x_0}\rangle$ from $|\psi_0\rangle$ in $O(n^2)$ operations. We can then use our construction of $\bs{R}$ and $\bs{t}$ in~\cref{eq_pauliDistributionK} to yield the distribution of $K$ which then allows us to compute the approximation to the frame potential in \cref{eq:tildeFU}. 

\section{An Example Application} \label{sec:example}

Suppose we set $N = n = 5$ and $U$ defined by \cref{eq:defU} with the following set of pairwise commuting Pauli operators,
\begin{equation}\label{eq:exampleHs}
    \begin{array}{lcc}
H_1 &=& -XXYYY \\
H_2 &=& IYIIX \\
H_3 &=& -IZXXZ \\
H_4 &=& XYIZI \\
H_5 &=& -XZXYY
\end{array}.
\end{equation}
Two Pauli operators commute if and only if the number of positions in which the strings differ and neither element is $I$ is even~\cite{sarkar2021sets}, which can be visually confirmed above. After simultaneous diagonalization, we obtain the unitary $W$ shown in circuit representation in \cref{fig:Wcircuit}.
\begin{figure}
    \centering
    \includegraphics[scale=0.8]{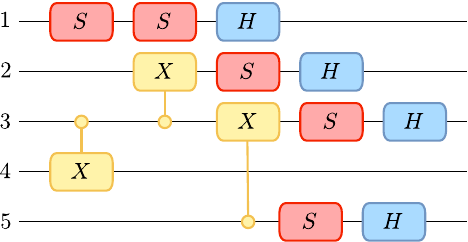}
    \caption{Circuit corresponding to the operators in \cref{eq:exampleHs}.}
    \label{fig:Wcircuit}
\end{figure}
After diagonalization, the diagonal operators in \cref{eq:exampleHs} can be encoded in the matrix
\[
\bs{A} = \left(\begin{tabular}{ccccc}
    1 & 0 & 0 & 1 & 1 \\
    0 & 1 & 1 & 0 & 1 \\ 
    0 & 1 & 1 & 0 & 0 \\ 
    1 & 1 & 0 & 1 & 0 \\ 
    1 & 1 & 0 & 1 & 1 
\end{tabular}\right).
\]

If we are to apply the simulation procedure in~\cite{Aaronson2004ImprovedSO}, the last $n$ rows of the tableau $\tau(|\psi_0\rangle)$ will be given by 
\[
\left(
\begin{tabular}
[c]{ccccc|ccccc|c}%
1 & 0 & 0 & 0 & 0 & 0 & 0 & 0 & 0 & 0 & 0\\
0 & 1 & 1 & 0 & 1 & 0 & 0 & 0 & 0 & 0 & 0\\
0 & 0 & 1 & 0 & 1 & 0 & 0 & 0 & 0 & 0 & 0\\
0 & 0 & 1 & 0 & 1 & 0 & 0 & 0 & 1 & 0 & 0\\
0 & 0 & 0 & 0 & 1 & 0 & 0 & 0 & 0 & 0 & 0\\
\end{tabular}\right). 
\]
from which we get $r=4$ and
\[ \bs{R} = \left(\begin{tabular}{cccc}
     1 & 0 & 0 & 0 \\
     0 & 1 & 0 & 0 \\ 
     0 & 0 & 1 & 0 \\ 
     0 & 0 & 0 & 0 \\ 
     0 & 0 & 0 & 1
\end{tabular}\right).  \]
Also, $|\bs{u}_{x_0} \rangle = |0\rangle^{\otimes n}$ is a measurable outcome from $|\psi_0\rangle$. Therefore, from \cref{prop:pmfK2} and \cref{thm:JK}, the probability mass function of $K$ is
\[ 
\mathbb{P}(K = (-1)^{\bs{b}}) = \frac{1}{16}, \quad \bs{b} \in \operatorname{range}(\bs{AR}).
\]
The resulting covariance matrix $\Cov(K) = \frac{1}{16} \bs{K}^{\intercal}\bs{K}$ is the identity matrix. The volume of the lattice can be computed as in  Prop. 2 of \cite{ramirez2024expressiveness}, and gives $V_U = 64$. As per \cref{eq:tildeFU}, the approximate frame potential for this circuit is
$$\tilde{\mathcal{F}}_U(t) = 2(\pi t)^{-5/2}.$$

Consider, in contrast, a circuit with $H_1=YZZIX$, $H_2=YYXII$, $H_3=-ZIYIX$, $H_4=ZXXXY$, $H_5=ZIYZI$, which corresponds to a set of Pauli strings in a different subgroup of $\pauli_n$ than the previous example. A similar calculation yields a matrix $\bs{R}$ of full rank and the  volume equal to the minimum possible value $V_U = 32$. Thus, frame potential is half that of the previous example and the circuit twice as expressive.

\section{Extending to General Sets of Pauli Operators}\label{sec:noncommutative}

The quantum expectation $f_U (\btheta)$ in \cref{eq:deffU} is a characteristic function of a random variable. However, it was never analyzed in~\cite{ramirez2024expressiveness} if $F_U (\btheta, \btheta')$ can also be a candidate for such a construction. To consider if this is possible, we still assume that $U(\btheta)$ is as in \cref{eq:defU} but we relax the assumptions we made before such that $\{H_j\}_{j = 1}^N$ can be any collection of Pauli operators that need not be commutative.  

We know that $F_U (\bs{0}, \bs{0}) = 1$ since the parameterized unitaries simply become the identity. Further, by its definition, $F_U (\btheta, \btheta')$ is nonnegative. Therefore, if we consider any finite sequence $\{(\btheta(k), \btheta'(k)\}_{k=1}^N$ of inputs, then we recognize that the matrix $\bs{F}$ defined by
\[
F_a^{(b)} = F_U (\btheta(a) - \btheta(b), \btheta'(a) - \btheta'(b))
\]
is a totally positive matrix. This means that if $F_U$ is an even function, that is, $F_U(\btheta, \btheta') = F_U(-\btheta, -\btheta')$, then $F_U$ is a positive definite function. This meets all the conditions of Bochner's criterion which states that $F_U$ will then be the characteristic function of some random variable. 

When calculating the frame potential $\mathcal{F}_U (t)$, let $F_U (\btheta, \btheta') = E_U (\btheta, \btheta') + O_U(\btheta, \btheta')$ where 
\[
E_U (\btheta, \btheta') = \frac{F_U (\btheta, \btheta') + F_U (-\btheta, -\btheta')}{2},
\]
\[
O_U (\btheta, \btheta') = \frac{F_U (\btheta, \btheta') - F_U (-\btheta, -\btheta')}{2}.
\]
It is easy to see that $E_U$ and $O_U$ are even and odd functions respectively. Then, for positive integer $t$, we have 

$$\mathcal{F}_U (t)= \frac{1}{(2 \pi)^{2 N}} \int_{[-\pi, \pi]^{2 N}} \sum_{k = 0}^{\lfloor t / 2 \rfloor} \binom{t}{2 k} O_U^{2 k} E_U^{t - 2k} \ud \btheta \ud \btheta'.$$

Owing to the antisymmetry of odd functions about the origin, the integrals of the odd component of $\mathcal{F}_U (t)$ disappears. The remaining integrand is now an even function for which a characteristic function exists. We observe $O_U^2 + 1$ and $E_U$ are themselves characteristic functions by Bochner's criterion. In order to fully characterize $F_U(\btheta, \btheta')$, we would need to identify its Fourier coefficients. These coefficients can be calculated using multiplicative sequences of Pauli operator projections. We have demonstrated that the simulation of stabilizer circuits gave us a concise description of the characteristic function of $f_U (\btheta, \btheta')$. Using a generalized version of this simulation procedure which encodes phases, we can then obtain representations of the Fourier series of $F_U (\bs{\theta}, \bs{\theta}')$ as well. This has been left as future work for us as well as any interested reader.

\section{Conclusion}\label{sec:conclusion}

In this paper, we have given a method which produces a nice closed form of the Fourier coefficients for parametric quantum circuits composed of commuting Pauli rotations. We extended the initial strategy developed in~\cite{ramirez2024expressiveness}, and showed how we can use an algorithmic approach, as given in~\cite{vandenBerg2020circuitoptimization} and~\cite{Aaronson2004ImprovedSO}, to produce such a formula. Further, our investigations show that variants of our approach can potentially be successful in evaluating the frame potential of circuits where a \textit{noncommuting} set of Pauli operators is allowed.  

Another intriguing avenue of study that has emerged through this line of work: given a parametric quantum circuit described by the unitary $U(\btheta) = \prod_{j = 1}^N \exp(i \theta_j H_j)$, what would be its Fourier expansion? We had shown already that stabilizers can enable such an answer, and that such a procedure could potentially be generalized to any set of Pauli operators. Using Suzuki-Trotter decompositions of higher order, there could be a possibility where we can accurately approximate the Fourier coefficients for any possible set of Hamiltonians $H_j$. 

This exploration of the Fourier series is also of interest as such a representation can be used to explore periodic behaviors of PQC architectures, which shows up in quantum machine learning applications. Further, it can have practical uses in the computation of certain quantities like the frame potential. Understanding the Fourier representation of the action of $U(\btheta)$ would give us the means to attack problems using tools that are traditionally not available to us, as demonstrated with the characteristic function of random variables. We anticipate that the ability to effectively compute such quantities by leveraging their underlying structural properties may play a future role in algorithmic and software design.

\section{Acknowledgements}
This work was completed under the aegis of the DOE Science Undergraduate Laboratory Internship (SULI) program and the Graduate Research at ORNL (GRO) program. Support for this work partially came from the U.S. Department of Energy's Advanced Scientific Computing Research (ASCR) Accelerated Research in Quantum Computing (ARQC) Program under the FAR-QC and AIDE-QC (field work proposal ERKJ332). We would also like to thank Ryan Bennink for providing us resources on how phase sensitive Clifford circuit simulation can be leveraged to find the Fourier coefficients of PQCs in the noncommutative scenario.

\bibliography{main}
\bibliographystyle{unsrt}

\end{document}